\documentclass{article}
\usepackage[latin1]{inputenc}
\usepackage[english]{babel}
\usepackage{amsfonts}
\usepackage{amssymb}
\usepackage{amsmath}
\usepackage{latexsym}
\newtheorem{theorem}{Theorem}
\newtheorem{lemma}[theorem]{Lemma}
\newtheorem{corollary}[theorem]{Corollary}
\newenvironment{proof}{\emph{Proof} : }{\begin{flushright}$\Box$\end{flushright}}
\newenvironment{remark}{\textbf{\emph{Remark}} : }{\vspace{0.2cm}}
\newenvironment{acknowledgements}{\textbf{\emph{Acknowledgements}} : }{\vspace{0.2cm}}
\newcommand{\F}[1]{\mathbb{F}_{#1}}

\begin{document}

\title{A new proof of Delsarte, Goethals and Mac Williams theorem on minimal weight codewords of generalized Reed-Muller codes}

\author{Elodie LEDUCQ}

\maketitle

\begin{abstract}
We give a new proof of Delsarte, Goethals and Mac williams theorem on minimal weight codewords of generalized Reed-Muller codes published in 1970. To prove this theorem, we consider intersection of support of minimal weight codewords with affine hyperplanes and we proceed by recursion.
\end{abstract}

\section{Introduction}
In the appendix of \cite{DGMW}, Delsarte, Goethals and Mac Williams prove the theorem \ref{min} below. However, at the beginning of their proof, they point out that "it would be very desirable to find a more sophisticated and shorter proof".
\\In this paper, we give a new proof of this theorem that we hope is simpler.
\\\\Let $q=p^n$, $p$ being prime number.\\ We identify the $\F{q}$-algebra $B_m^q=\F{q}[X_1,\ldots,X_m]/(X_1^q-X_1,\ldots,X_m^q-X_m)$ to the $\F{q}$-algebra of  the functions from $\F{q}^m$ to $\F{q}$ through the isomorphism $P\mapsto (x\mapsto P(x))$.
\\For $f\in B^m_q$, let $S_f=\{x\in\F{q}^m, f(x)\neq0\}$ the support of $f$, and \\$|f|=\mathrm{Card}(S_f)$  the weight of $f$. The Hamming distance in $B_m^q$ is denoted by $d(.,.)$.
\\For $0\leq r\leq m(q-1)$, the $r$th order generalized Reed-Muller code of length $q^m$ is $$R_q(r,m)=\{P\in B_m^q, \mathrm{deg}(P)\leq r\}$$ where $\mathrm{deg}(P)$ is the degree of the representative of $P$ with degree at most $q-1$ in each variable.
\\The affine group $\mathrm{GA}_m(\F{q})$ acts on $R_q(r,m)$ by its natural action. The minimum weight of $R_q(r,m)$ is $(q-s)q^{m-t-1}$, where $r=t(q-1)+s$, $0\leq s\leq q-2$ (see \cite{DGMW}).
\\\\The following theorem gives the codeword of minimum weight of $R_q(r,m)$

\begin{theorem}\label{min}Let $r=t(q-1)+s<m(q-1)$. The minimal weight codewords of $R_q(r,m)$ are codewords of $R_q(r,m)$ whose support is the union of $(q-s)$ distinct parallel affine subspaces of codimension $t+1$ included in an affine subspace of codimension $t$. \end{theorem}

\begin{remark} \begin{enumerate}\item Clearly, codewords of this form are of minimal weight. \item Using lemma \ref{2} and corollary \ref{3} below, this theorem means that codewords of minimal weight are equivalent, under the action of the affine group, to a codeword of the following form:
$f(x)=c\displaystyle\prod_{i=1}^t(x_i^{q-1}-1)\prod_{j=1}^s(x_{t+1}-b_j)$
where $c\in\F{q}^*$, $b_j$ are distinct elements of $\F{q}$.\end{enumerate}\end{remark}

\section{Proof of Theorem \ref{min}}

In this paper, we use freely the two following lemmas and their corollary proved in \cite{DGMW} p. 435.

\begin{lemma}\label{2} If $P(x)=0$ whenever $x_1=a$, then $P(x)=(x_1-a)Q(x)$ where $\mathrm{deg}_{x_1}(Q)\leq \mathrm{deg}_{x_1}(P)-1$.
\end{lemma}

\begin{remark}In \cite{DGMW}, the lemma A1.1 says that the exponent of $x_1$ in $Q$ is at most $q-2$, but they actually prove the above lemma.\end{remark}

\begin{corollary}\label{3}If $P(x)=0$ unless $x_1=b$, then $P(x)=(1-(x_1-b)^{q-1})Q(x_2,\ldots,x_m)$.\end{corollary}

\begin{lemma}\label{4} Let $S$ be a subset of $\F{q}^m$, such that $\mathrm{Card}(S)=tq^n<q^m$, $0<t<q$.
\\Assume that for any hyperplane of $\F{q}^m$, either $\mathrm{Card}(S\cap H)=0$ or \\$\mathrm{Card}(S\cap H)\geq tq^{n-1}$. Then there exists an affine hyperplane of $\F{q}^m$ which does not meet $S$.\end{lemma}

Now, we prove a lemma that will be crucial in our proof of theorem \ref{min} :

\begin{lemma}\label{5}Let $r=t(q-1)+s$, $0\leq s\leq q-2$, $f$ be a minimal weight codeword of $R_q(r,m)$ and $S=S_f$.
\\If $H$ is an hyperplane of $\F{q}^m$, such that $S\cap H\neq\emptyset$ and $S\cap H\neq S$, then
either $S$ meets all hyperplanes parallel to $H$ or $S$ meets $q-s$ hyperplanes parallel to $H$ in $q^{m-t-1}$ points.\end{lemma}

\begin{proof}Since an affine transformation does not change weight, we can assume that $H=\{x,x_1=0\}$.
\\Now assume that $S$ does not meet $k$ hyperplanes parallel to $H$, $k\geq1$. As $S\cap H\neq\emptyset$ and $S\cap H\neq H$, we have $k\leq q-2$. By lemma \ref{2}, we can write $$f(x)=(x_1-b_1)^{\alpha_1}\ldots(x_1-b_k)^{\alpha_k}P(x)$$ where $S_P$ meets all hyperplanes parallel to $H$.
\\Let $d=\displaystyle\sum_{i=1}^k\alpha_i\leq q-1$. We want to prove that $d=s$.
\begin{itemize}\item First, assume that $d>s$. Then the degree of $P$ is $(t-1)(q-1)+q-1+s-d$ and $0\leq q-1+s-d\leq q-2$. For $c\not\in\{b_1,\ldots,b_k\}$, we consider \\$Q_c=(1-(x_1-c)^{q-1})P$. The degree of $Q_c$ is $t(q-1)+q-1+s-d$. So
\begin{eqnarray*}(q-s)q^{m-t-1}=\mathrm{Card}(S)&=&\sum_{c\not\in\{b_1,\ldots,b_k\}}\mathrm{Card}(S_P\cap\{x_1=c\})\\&\geq&(q-k)(q-(q-1+s-d))q^{m-t-1}\\&=&(q-k)(d-s+1)q^{m-t-1}\end{eqnarray*}
and we obtain, since $d\geq k$,
$$(q-1-k)(d-s)\leq 0,$$
which is impossible, since $d>s$ and $k<q-1$.
\item Now we have $d\leq s$.
\\So $\mathrm{deg}(P)=t(q-1)+s-d$ and we get
\begin{eqnarray*} (q-s)q^{m-t-1}=\mathrm{Card}(S)&=&\sum_{c\not\in\{b_1,\ldots,b_k\}}\mathrm{Card}(S_P\cap\{x_1=c\})\\&\geq&(q-k)(q-s+d)q^{m-t-2}\end{eqnarray*}
which gives $$(d-k)q+k(s-d)\leq 0.$$
Hence, since $d\geq k$, $k\geq1$ and $s\geq d$, necessarily $d=k=s$ \\and $\mathrm{Card}(S_P\cap\{x_1=c\})=q^{m-t-1}$ for $c\not\in\{b_1,\ldots,b_k\}$.\end{itemize}\end{proof}

Now, we are able to prove theorem \ref{min}.\\

\begin{proof}
We prove first the case where $t=0$ and $t=m-1$.
\\\\$\bullet$ \quad  $t=0$.
\\If $s=0$, then $\mathrm{deg}(f)=0$. Thus, since $f\neq0$, we have $f=c$, for $c\in\F{q}^*$ and $S_f=\F{q}^m$.
\\Otherwise, let $H$ be an affine hyperplane of $\F{q}^m$, then $\mathrm{Card}(S_f\cap H)=0$ or $\mathrm{Card}(S_f\cap H)\geq (q-s)q^{m-2}$.
\\Hence, by lemma \ref{4}, there exists an affine hyperplane $H_0$, such that $S_f\cap H_0=\emptyset$.
\\However, $\F{q}^m$ is the union of the $q$ hyperplanes parallel to $H_0$, so there exists $H_1$, parallel to $H_0$, such that $H_1\cap S_f\neq\emptyset$.
\\Furthermore, since $|f|=(q-s)q^{m-1}\geq2q^{m-1}$, $S_f\cap H_1\neq S_f$. So, by lemma \ref{5}, since $S_f\cap H_0=\emptyset$, $S_f$ meets $(q-s)$ hyperplanes parallel to $H_1$, say $H_1,\ldots,H_{q-s}$, in $q^{m-1}$ points, this means that $S_f=\displaystyle\bigcup_{i=1}^{q-s}H_i$.
\\\\$\bullet \quad\ t=m-1$.
\\Let $f\in R_q((m-1)(q-1)+s,m)$, $0\leq s\leq q-2$, such that $|f|=q-s$. We put $S=S_f$. Let $\omega_1$, $\omega_2\in S$
and $H$ be an hyperplane, such that $\omega_1$, $\omega_2\in H$.
\\Assume that $S\cap H\neq S$ then , by lemma \ref{5}, either $S$ meets all hyperplanes parallel to $H$ (which is possible  only if $s=0$) or $S$ meets $(q-s)$ hyperplanes parallel to $H$ in one point. In both cases we get a contradiction, since in both cases $S$ meets each hyperplane in exactly one point and $\omega_1$, $\omega_2\in H$.
\\So $S$ is included in all hyperplanes $H$, such that $\omega_1$, $\omega_2\in H$, this means that $S$ is included in the line through $\omega_1$ and $\omega_2$.
\\\\Now we prove the theorem for general $t$ by recursion.
\\\\$\bullet$\quad Assume that for a fixed $t$, $1\leq t\leq m-2$ (we have already proved the case where $t=0$) and for all $0\leq s\leq q-2$, the support of a codeword of minimal weight in $R_q((t+1)(q-1)+s,m)$ is the union of $(q-s)$ distinct parallel affine subspaces of codimension $t+2$ included in an affine subspace of codimension $t+1$.
\\\\Let $f\in R_q(t(q-1)+s,m)$, such that $|f|=(q-s)q^{m-t-1}$.
\\We put $S=S_f$. Let $a \in S$ and $F=\{\overrightarrow{ab},b\in S\}$. We have :
$$\mathrm{Card}(F)=(q-s)q^{m-t-1}\leq q^{rg(F)}.$$ Thus, since $0\leq s\leq q-2$, we have $rg(F)\geq m-t$.
\\Let $\overrightarrow{v_1},\ldots,\overrightarrow{v_{m-t}}$ be $m-t$ independent vectors of $F$ and $\overrightarrow{u}$, such that \\ $\overrightarrow{u}\not\in \mathrm{Vect}(\overrightarrow{v_1},\ldots,\overrightarrow{v_{m-t}})$.
\\Since $t\geq 1$, there exists an affine hyperplane, say $H$, such that\\ $a+\overrightarrow{v_1},\ldots,a+\overrightarrow{v_{m-t}}\in H$ and $a+\overrightarrow{u}\not\in H$.
\\Assume that $S\cap H\neq S$. Then by lemma \ref{5}, either $S$ meets all hyperplanes parallel to $H$ or $S$ meets $(q-s)$ hyperplanes parallel to $H$ in $q^{m-t-1}$ points.
\begin{itemize}\item\emph{$1$st case} : $S$ meets $(q-s)$ hyperplanes.
\\By applying an affine transformation, we can assume that the $q-s$ hyperplanes are $H_i=\{x,x_1=a_i\}$, $a_i\in \F{q}$. Without loss of generality, we can assume that $H=H_1$.
\\Let $P=\displaystyle\prod_{i=2}^{q-s}(x_1-a_i)f(x)$, \\ $\mathrm{deg}(P)\leq t(q-1)+s+q-1-s=(t+1)(q-1)$ and \\ $|P|=\mathrm{Card}(S\cap\{x_1=a_1\})=q^{m-t-1}$. So $P$ is a codeword of minimal weight in $R_q((t+1)(q-1),m)$, and, by recursion hypothesis, $S_p=S\cap H$ is an affine subspace of codimension $t+1$.
\item\emph{$2$nd case} : $S$ meets all the hyperplanes.
\\For all $G_a$ hyperplane of equation $(z=a)$, $a\in\F{q}$, parallel to $H=G_0$, \\$g_a=f.(1-(z-a)^{q-1})\in R_q((t+1)(q-1)+s,m)$ and $g_a\neq0$. So $\mathrm{Card}(S_{g_a})\geq(q-s)q^{m-t-2}$. \\Since $\mathrm{Card}(S)=(q-s)q^{m-t-1}$ and $S_{g_a}=S\cap G_a$, \\$\mathrm{Card}(S_{g_a})=(q-s)q^{m-t-2}$.
\\By recursion hypothesis, $S_{g_0}=S\cap H$ is included in an affine subspace of codimension $t+1$.\end{itemize}
In both cases, $S\cap H$ is included in  an affine subspace of codimension $t+1$ which is impossible since $a+\overrightarrow{v_1},\ldots,a+\overrightarrow{v_{m-t}}\in S\cap H$. \\So $S\cap H=S$ and $a+\overrightarrow{u}\not\in S$, which means that $\overrightarrow{u}\not\in F$.\\ Hence, $F\subset\mathrm{Vect}(\overrightarrow{v_1},\ldots,\overrightarrow{v_{m-t}})$, \emph{i.e} $S$ is included in an affine subspace of codimension $t$, say A.
\\\\By applying an affine transformation, we can assume that\\ $A=\{x,x_{1}=0,\ldots,x_t=0\}$. Then by corollary \ref{3}, we can write
$$f(x)=\prod_{i=1}^t(x_i^{q-1}-1)P(x_{t+1},\ldots,x_{m})$$
$P\in R_q(s,m-t)$ and $|P|=|f|=(q-s)q^{m-t-1}$, thus, by the case where $t=0$, $S_P$ is the union of $(q-s)$ parallel hyperplanes of $A$ which gives the result.
\end{proof}
\begin{acknowledgements}
I want to thank my supervisor, Jean-François Mestre, for his very helpful remarks.
\end{acknowledgements}

\end{document}